\documentclass[11pt, a4paper]{article}

\usepackage[T1]{fontenc}

\usepackage{amsmath}
\usepackage{enumitem}
\usepackage{amssymb}
\usepackage{amsfonts}
\usepackage{textcomp}
\usepackage{algorithm}
\usepackage{authblk}
\usepackage[noend]{algpseudocode}
\usepackage{amsthm}
\usepackage{xspace}
\usepackage{fullpage}
\usepackage[english]{babel}
\usepackage[utf8]{inputenc}
\usepackage{mathtools}
\usepackage{hyperref}
\usepackage[dvipsnames]{xcolor}
\hypersetup{
	colorlinks=true,
	urlcolor=MidnightBlue,
	linkcolor=MidnightBlue,
	citecolor=MidnightBlue,
	unicode
}

\usepackage{subfig}
\usepackage{todonotes}
\usepackage{thmtools}
\usepackage{thm-restate}
\usepackage[capitalise, noabbrev, nameinlink]{cleveref}
\usepackage{bm}
\usepackage{tcolorbox}

\theoremstyle{plain}
\newtheorem{theorem}{Theorem}
\newtheorem{lemma}[theorem]{Lemma}

\newtheorem{proposition}[theorem]{Proposition}
\newtheorem{claim}[theorem]{Claim}

\theoremstyle{definition}

\newtheorem{definition}[theorem]{Definition}

\tcbuselibrary{breakable}
\tcbset{
  width=1.02\textwidth,
  left=3mm,
  right=3mm,
  halign=justify,
  center,
  breakable,
  colback=white
}

\newcommand{\newauthor}[3]{
    \newcounter{#1comment}
    \setcounter{#1comment}{1}
    \expandafter\newcommand\csname #1\endcsname[1]{%
            \par\noindent
            \todo[inline, size = \small, backgroundcolor = {#3}, caption = {}]{
                \arabic{#1comment}:
                {##1} --~\textbf{#2}
            }
            \addtocounter{#1comment}{1}
    }
    \expandafter\newcommand\csname #1changed\endcsname[1]{%
        \ifdraft
            \colorbox{#3}{
                ##1
            }%
        \else
            ##1%
        \fi
    }
}

\newauthor{ds}{Dmitry}{Cyan}
\newauthor{kr}{Kilian}{LimeGreen}
\newauthor{gm}{Gilbert}{Orange}

\newcommand{\set}[1]{\{ #1 \}}

\newcommand{\eqperiod}{\enspace .}
\newcommand{\eqcomma}{\enspace ,}
\newcommand{\ie}{i.e.,\ }

\newcommand{\floor}[1]{\lfloor #1 \rfloor}

\newcommand{\eps}{\varepsilon}

\newcommand{\calC}{\mathcal{C}}

\newcommand{\calP}{\mathcal{P}}

\newcommand{\N}{\mathbb{N}}

\newcommand{\red}{\color{BrickRed}}
\newcommand{\blue}{\color{MidnightBlue}}
\newcommand{\gray}{\color{Gray!30!white}}

\newcommand{\textb}[1]{\texttt{\upshape #1}}

\newcommand{\lred}{{\red\textb{[}}}
\newcommand{\rred}{{\red\textb{]}}}
\newcommand{\tred}{{\red\textb{[\!]}}}

\newcommand{\lblue}{{\blue\textb{[}}}
\newcommand{\rblue}{{\blue\textb{]}}}
\newcommand{\tblue}{{\blue\textb{[\!]}}}

\newcommand{\Bracket}[1][n]{\mathrm{WideBr}_{#1}}
\newcommand{\ExBracket}[1][n]{\mathrm{Br}_{#1}}

\newcommand{\Ind}{\textsc{Ind}}

\DeclareMathOperator{\ancestors}{ancestors}

\DeclareMathOperator{\area}{area}

\DeclareMathOperator{\poly}{poly}

\DeclareMathOperator{\supp}{supp}

\crefname{step}{Step}{Steps}
\Crefname{step}{Step}{Steps}

\crefname{prop}{Property}{Properties}
\Crefname{prop}{Property}{Properties}

\newcommand{\indexBeforeCall}{{i}}
\newcommand{\indexAfterCall}{{i+1}}

\begin{document}

\mbox{}\vspace{12mm}

\begin{center}
{\huge Supercritical Tradeoffs for Monotone Circuits}
\\[1cm] \large

\setlength\tabcolsep{1.5em}
\begin{tabular}{cccc}
Mika Göös&
Gilbert Maystre&
Kilian Risse&
Dmitry Sokolov\\[-1mm]
\small\slshape EPFL&
\small\slshape EPFL&
\small\slshape EPFL&
\small\slshape EPFL
\end{tabular}
	
\vspace{2em}
	
\large
{\today}

\normalsize

\vspace{1em}
\end{center}
	
\begin{abstract}
\noindent
We exhibit a monotone function computable by a monotone circuit of quasipolynomial size such that any monotone circuit of polynomial depth requires exponential size. This is the first size--depth tradeoff result for monotone circuits in the so-called \emph{supercritical} regime. Our proof is based on an analogous result in proof complexity: We introduce a new family of unsatisfiable~3-CNF formulas (called \emph{bracket formulas}) that admit resolution refutations of quasi\-polynomial size while any refutation of polynomial depth requires exponential size.
\end{abstract}

\section{Introduction}

In monotone circuit complexity, a \emph{size--depth tradeoff} result states that monotone circuits cannot
be compressed to shallow depth without blowing up their size. (See \cref{sec:related-work-circuits} for
examples from prior work.) Our main result is such a size--depth tradeoff in the so-called
\emph{supercritical} regime of parameters. This type of tradeoff was first envisioned
by~\cite{Beame16,Razborov16} and refers to a blow-up in one parameter~$S$ (for us, size) when another
parameter~$D$ (for us, depth) is restricted to a value above its \emph{critical value}, defined as the
largest value that $D$ can take when maximised over all $n$-bit functions (and when~$S$ is
unrestricted). Concretely, if we consider monotone circuits of fanin~2, the worst-case upper bound for
depth is~$D\leq\Theta(n)$: any $n$-bit monotone function can be computed by a monotone circuit of depth
$n$, and some (e.g., random) functions require depth~$\Omega(n)$. We show that restricting the depth to
any polynomial bound,~$D\leq n^c$ where~$c>1$ is any constant, can still cause a huge blow-up in size.

\begin{restatable}[Main result for circuits]{theorem}{TheoremMainCkt}
    \label{thm:main-circuit}
    There is a monotone $f\colon \{0, 1\}^n \to \{0, 1\}$ such that:
    \begin{enumerate}[label = {\itshape (\roman*)}, noitemsep]
        \item \label{it:ckt-ub}
            There is a monotone circuit computing $f$ in size $n^{O(\log n)}$.
        \item \label{it:ckt-lb}
            Every monotone circuit computing $f$ in depth $n^{O(1)}$ requires size $\exp(n^{\Omega(1)})$.
    \end{enumerate}
\end{restatable}

Supercritical tradeoffs for monotone circuits were conjectured by~\cite{Garg20,Fleming21,Fleming22}. Our \cref{thm:main-circuit} does not yet, however, quantitatively reach the parameters stated in those conjectures. They asked for a blow-up to occur already for depth $D\leq S^\varepsilon$ where $\varepsilon>0$ is a constant and $S$ is the size of the circuit in~\ref{it:ckt-ub}. We leave it for future work to quantitatively sharpen our tradeoff. Nevertheless, our result is the first supercritical tradeoff for monotone circuits (although see~\cref{sec:concurrent-work} for a discussion of a concurrent work~\cite{Rezende24}).

Our proof follows the by-now standard \emph{lifting paradigm}: We start by proving an analogous tradeoff
result in propositional proof complexity, and then apply a lifting theorem~\cite{Garg20} to import that
result to monotone circuit complexity. We state our proof complexity tradeoff for the standard
\emph{resolution} proof system; see the textbook~\cite[\S18]{Jukna12} for an introduction.

\begin{restatable}[Main result for proofs]{theorem}{TheoremMainProof}
    \label{thm:main-proof}
    There is an $n$-variate $3$-CNF formula $F$ such that:
    \begin{enumerate}[label={\itshape (\roman*')},noitemsep]
        \item \label{it:res-ub}
            There is a resolution refutation of $F$ in size $n^{O(\log n)}$.
        \item \label{it:res-lb}
            Every resolution refutation of $F$ in depth $n^{O(1)}$ requires size $\exp(n^{\Omega(1)})$.
    \end{enumerate}
\end{restatable}

The previous best size--depth tradeoffs for resolution were due to Fleming, Robere, and Pitassi~\cite{Fleming22} with a more direct proof given by Buss and Thapen~\cite{Buss24}. They obtain a tradeoff that is able to match (and go beyond) all our parameters in~\ref{it:res-ub}--\ref{it:res-lb}, but only for a~CNF formula~$F$ having quasipolynomially many clauses. That is, their result does not apply when the depth restriction is superlinear in the size of $F$. As the authors discuss~\cite[\S7]{Fleming22}, this limitation prevents their results from lifting to a monotone circuit tradeoff. The question of overcoming this limitation has been asked several times previously~\cite{Razborov16,Berkholz20,Fleming22}. Our result now addresses this issue, since our formula is a~{3-CNF} and hence of size~$O(n^3)$. (The same limitation was also overcome in the concurrent work discussed in~\cref{sec:concurrent-work}.)

\subsection{New technique: Bracket formulas}

To prove \cref{thm:main-proof} we introduce a family of CNF formulas based on a novel \emph{bracket principle}. Consider a string $s \in \set{\lred, \lblue, \rred, \rblue}^n$ of red and blue brackets that begins with a red opening bracket, $s_1=\lred$, ends with a blue closing bracket, $s_n=\rblue$, and is \emph{well-parenthesised}: every opening bracket can be paired with a subsequent closing bracket of the same colour, and the bracket pairs are correctly nested (forbidding the pattern $\lred\cdots\lblue\cdots\rred\cdots\rblue$). For example,
\[
s ~=~
\lred\lblue\lred\rred\rblue\lred\rred\rred\lred\rred\lred\lblue\lred\lblue\rblue\lblue\rblue\rred\rblue\rred\lblue\lred\lblue\lblue\rblue\rblue\rred\rblue\lblue\lblue\rblue\lblue\rblue\rblue.
\]
We claim that every such string necessarily contains at least one occurrence of the substring $\rred\lblue$. To see this, we can erase all but the top-level brackets to obtain the string
\[
\hphantom{s ~=~ }~
\lred\hphantom{\lred\lblue\lred\rred\rblue\rred}\rred\lred\rred\lred\hphantom{\lblue\lred\lblue\rblue\lblue\rblue\rred\rblue}\rred\lblue\hphantom{\lred\lblue\lblue\rblue\rblue\rred}\rblue\lblue\hphantom{\lblue\rblue\lblue\rblue}\rblue.
\]
Here it is clear that the region of red brackets has to transition over to the blue region, which is where we find the substring $\rred\lblue$.

\paragraph{CNF encoding with pointers.}
We can now encode the bracket principle as an unsatisfiable \emph{bracket formula}, denoted $\ExBracket$. That is, we express as a CNF the negation of the tautology
\begin{equation}\label{eq:taut}
(s_1=\lred) \land  (s_n=\rblue) \land (s\text{ is well-parenthesised}) \quad\Longrightarrow\quad \exists i\in[n]\colon s_is_{i+1}=\rred\lblue.
\end{equation}
To help express the property that $s$ is well-parenthesised, we equip each bracket with a pointer to its mate. That is, we actually consider strings over the extended alphabet $\Sigma\coloneqq \set{\lred, \lblue, \rred, \rblue}\times [n]$. For a string $s\in\Sigma^n$ to be well-parenthesised, we require that the pointers encode a correct pairing of the brackets. For example, if $\smash{s_i=(\lred,j)}$, then we require that $s_j=(\rred,i)$ and $\smash{j>i}$. The correct nesting property can be expressed by forbidding patterns for every 4-tuple of indices. Finally, if we encode symbols in the alphabet as $O(\log n)$-bit strings, it is straightforward to write down a $\poly(n)$-size $O(\log n)$-width CNF for $\neg$\eqref{eq:taut}. The width can be further reduced to~3 by using standard encoding tricks. The full formal definition of $\ExBracket$ is given in~\cref{sec:bracket-formula}. 

\paragraph{Proof overview.}
We use the bracket formulas $\ExBracket$ to prove~\cref{thm:main-proof}. It is a recurring theme in proof complexity that the size of a resolution refutation $\Pi$ is often closely related to its \emph{width}, defined as the maximum width of a clause appearing in $\Pi$. Namely, any refutation of width~$w$ has size at most $n^{O(w)}$. Conversely, a formula $F$ requiring resolution width $w$ can be lifted (e.g.,~\cite{Garg20}) to a related formula $F'$ requiring refutations of size $n^{\Omega(w)}$. (A somewhat weaker converse holds also without lifting in case $w$ is large enough~\cite{BenSasson01}.) We start by showing that the bracket formulas exhibit a \emph{width--depth} tradeoff according to the following two theorems. The proofs appear in \cref{sec:upper-bound,sec:lower-bound}, respectively.
\begin{restatable}{theorem}{TheoremUB}
    \label{thm:width-ub}
    The bracket formula $\ExBracket$ admits a resolution refutation of width $O(\log n)$.
\end{restatable}
\vspace{-1em}
\begin{restatable}{theorem}{TheoremLB}
    \label{thm:depth-lb}
    Every width-$w$ resolution refutation of $\ExBracket$ has
    depth~$\smash{n^{\Omega\bigl(\frac{\log n}{\log w}\bigr)}}$.
\end{restatable}

It is now straightforward to derive our main results (\cref{thm:main-circuit,thm:main-proof}) by a black-box application of the lifting theorem from~\cite{Garg20}, which converts the above width--depth tradeoffs into size--depth tradeoffs. In fact, we get analogous tradeoffs also for the \emph{cutting planes} proof system. The formal proof is given in \cref{sec:lifting}.

\paragraph{Genealogy of bracket formulas.}
Perhaps surprisingly, we did not construct bracket formulas initially for the purpose of proving a tradeoff result. The bracket principle arose serendipitously from our unsuccessful attempts at classifying the \emph{Jordan curve} (JC) theorem within the total search problem theory {\sffamily TFNP} (specifically, the \emph{crossing curves} variant defined by~\cite{Adler16}). While the complexity of JC remains wide open, we currently know~\cite{Hollender24} that the bracket principle reduces to JC, and the complexity of the bracket principle falls somewhere between the \emph{unique end-of-potential-line} class~{\sffamily UEOPL}~\cite{Fearnley2020} and the \emph{end-of-potential-line} class~$\textsf{EOPL}=\textsf{PLS}\cap\textsf{PPAD}$~\cite{Goos24collapse}. For example, \cref{thm:width-ub} can be interpreted as showing that the principle lies in the class {\sffamily PLS} (which is characterised by low-width resolution; see~\cite{Goos24tfnp} for a discussion). Proving the bracket principle complete for any existing {\sffamily TFNP} class is an interesting open problem.

\subsection{Related work: Monotone circuit complexity
} \label{sec:related-work-circuits}

We have claimed \cref{thm:main-circuit} as the first supercritical size--depth tradeoff for monotone circuits. If we are pedantic, however, and consider the model of \emph{unbounded-fanin} monotone circuits, then the critical value for depth becomes $2$: any monotone function can be computed by a monotone DNF, which has depth 2. Under this interpretation, classical results in monotone circuit complexity can indeed be viewed as supercritical tradeoffs for unbounded-fanin circuits. For a prominent example, a long line of work~\cite{Karchmer88,Raz99,Goos18,Rezende16,Rezende20} has culminated in an extreme separation of the monotone analogues of the classes {\sffamily NC} and {\sffamily P}. They exhibit an $n$-bit monotone function~$f$ computed by a linear-size monotone circuit such that any unbounded-fanin monotone circuit computing $f$ in depth $n^{1-\eps}$ (which is larger than the critical value $2$) requires size $\exp(n^{\Omega(1)})$. Similar blow-ups in size also occur for constant-depth monotone circuits, as is known since the 1980s~\cite{Klawe84} with modern results given by Rossman~\cite{Rossman15,Rossman24}.

Given that the notion of supercriticality was first conceived in the context of the resolution proof system~\cite{Beame16,Razborov16} (where the critical value for depth/space is $n$), in this paper, we have chosen to reserve the term \emph{supercritical} for when the critical value for depth is taken to be~$\Theta(n)$. This is also how the term was used in the papers conjecturing a supercritical tradeoff for monotone circuits~\cite{Garg20,Fleming21,Fleming22}. (Note that our \cref{thm:main-circuit} holds no matter what fanin we consider: when we convert unbounded-fanin circuits to bounded-fanin circuits, the depth can increase only by a factor of $\log(\text{fanin})\leq n$.).

\subsection{Related work: Proof complexity}
\label{sec:related-work-proofs}

In proof complexity, the first supercritical tradeoffs were proven by Beame, Beck, and Impagliazzo~\cite{Beame16} (resolution size--space) and Razborov~\cite{Razborov16} (tree-like resolution size/depth--width), with another early work by Berkholz~\cite{Berkholz12} (resolution depth--width). Since then, the phenomenon has been studied extensively for resolution space~\cite{Razborov17space,Berkholz20,Papamakarios23}, as well as for other proof systems such as polynomial calculus~\cite{Beck13}, cutting planes~\cite{Razborov17cp,Fleming22}, and tree-like resolution over parities~\cite{Chattopadhyay24}. As already discussed above, the state-of-the-art for resolution size--depth are given by~\cite{Fleming22,Buss24}.

\paragraph{Hardness condensation.}
In all results cited above (except~\cite{Berkholz12}), the basic technique to obtain supercritical tradeoffs is \emph{hardness condensation}. This technique proceeds as follows.
\begin{enumerate}[label=(\arabic*)]
\item \label[step]{it:hc1}
One starts with a $n$-variate formula $F$ that exhibits a \emph{sub}critical tradeoff between the two parameters of interest. Here, it suffices to consider standard formulas ubiquitous in proof complexity, such as \emph{Tseitin} (used by~\cite{Beame16,Razborov17space,Beck13}) or \emph{pebbling} (used by~\cite{Razborov16,Berkholz20,Fleming22,Buss24,Chattopadhyay24}).
\item \label[step]{it:hc2}
One then \emph{reduces the number of variables} of $F$ while \emph{preserving the original hardness}. Methods to reduce the number of variable include identifying groups of variables (projections) and composing the formula with small gadgets (e.g., {\footnotesize XOR} substitutions, pointers).
\item The resulting formula $F'$ now condenses the original hardness. The hardness parameter of~$F'$ is roughly the same as for $F$, but now the tradeoff has become supercritical: the parameters become larger as a function of the number of variables.
\end{enumerate}

In this work, we do not explicitly employ hardness condensation. Our bracket formulas exhibit robust tradeoffs \emph{naturally} right out of the box. It should be noted, however, that hardness condensation, broadly construed, is a \emph{complete} method for proving size--depth tradeoffs. Any formula witnessing a size--depth tradeoff can be obtained by starting with a pebbling formula in \cref{it:hc1} and then condensing using operations listed in \cref{it:hc2}.\footnote{This is because pebbling/sink-of-dag formulas are \emph{complete} for resolution: Any formula $F$ with a small resolution proof can be reduced, via decision trees, to a pebbling formula $P$. In other words, $F$ is obtained from~$P$ by identifying variables and local gadget composition. See~\cite{Goos24tfnp} for an exposition of this perspective.}

Resolution width--depth tradeoffs (from hardness condensation) are at the heart of recent breakthroughs in showing lower bounds on the stabilisation time of the Weisfeiler--Lehman algorithm for graph isomorphism~\cite{Berkholz23,Grohe23,Grohe23a}. Can our bracket formulas find applications in this line of work? Finally, we mention that hardness condensation has also been studied in circuit complexity~\cite{BureshOppenheim06} (where the technique was first proposed) and query/communication complexity~\cite{Goos24,Hrubes24}.

\subsection{Concurrent work by \texorpdfstring{\cite{Rezende24}}{dRF25}}
\label{sec:concurrent-work}

In concurrent work, de Rezende, Fleming, Janett, Nordstr{\"o}m, and Pang~\cite{Rezende24} have independently obtained a supercritical tradeoff for monotone circuits. They exhibit an $f$ so that:
\begin{enumerate}[label={\itshape (\roman*'')},noitemsep]
\item \label{it:other-ub}
There is a monotone circuit computing $f$ in size $n^c$ for some constant $c>1$.
\item \label{it:other-lb}
Every monotone circuit computing $f$ in depth $n^{1.9}$ requires size $n^{1.5c}$.
\end{enumerate}
Comparing this to our result, the size blow-up in \ref{it:other-lb} is at most polynomial, while our blow-up in \ref{it:ckt-lb} is nearly exponential. On the other hand, their blow-up occurs already at depth that is polynomial in the size upper bound~\ref{it:other-ub}. In this sense, the two results are incomparable.

Their approach, too, is to start with a new supercritical width--depth tradeoff for resolution (a condensed Tseitin formula, building on~\cite{Grohe23a}) with parameters analogous to \ref{it:other-ub}--\ref{it:other-lb} and then lift that to a monotone circuit tradeoff. Since they start with a small polynomial blow-up in proof complexity, they have to be careful that the lifting construction is not too lossy in its parameters. Consequently, they prove a new more optimised lifting theorem based on~\cite{Garg20,LovettMMPZ22}. In comparison, we can afford a black-box application of these theorems. The paper~\cite{Rezende24} also gives applications to the Weisfeiler--Lehman algorithm, strengthening the breakthrough result of~\cite{Grohe23a}. By contrast, we leave it as a direction for further research to find out if bracket formulas can yield such applications.

Some preliminary results of~\cite{Rezende24} were announced already at an Oberwolfach workshop in March 2024. Based on that announcement, Berkholz, Lichter, and Vinall{-}Smeeth~\cite{Berkholz24} have recently obtained new supercritical size--width tradeoffs for tree-like resolution.

\section{Bracket Formulas}
\label{sec:bracket-formula}

In this section, we formally define \emph{bracket formulas} as a particular 3-CNF encoding of the bracket
principle~\eqref{eq:taut}. For convenience, we start by first expressing bracket principles as a system
of constraints over a large alphabet, and then describe how to encode them in binary.

\subsection{Large-alphabet encoding}
We consider strings $s \in \Sigma^n$ over the alphabet
\[
    \Sigma \coloneqq \set{\tred, \tblue, \lred, \lblue, \rred, \rblue} \times [n].
\]
Here we have introduced \emph{trivial brackets} $\tred$ and $\tblue$ for technical convenience (it helps
us avoid parity issues, e.g., requiring $n$ to be even). They represent an open/close bracket pair,
$\lred\, \rred$ or~$\lblue\, \rblue$, compressed to a single symbol. For a symbol $\sigma \in \Sigma$ we
write $\sigma = (b(\sigma), p(\sigma))$ where $b(\sigma)$ is the bracket type and $p(\sigma) \in [n]$ is the
pointer. Given a string $s \in \Sigma^n$ we define the following set of constraints, called axioms. 
\begin{enumerate}[label=(A\arabic*), leftmargin = 3em, itemsep = 0pt]
    \item \label{a1}
        \emph{Start/end of string:} $b(s_1)\in\{\tred,\lred\}$ and $b(s_n)\in\{\tblue,\rblue\}$.
    \item \label{a2}
        \emph{Pointers define bracket pairs:}
        \begin{itemize}[leftmargin = 1em, label=$-$, noitemsep]
            \item If $p(s_i) = j$, then $p(s_j) = i$. We say $s_i$ and $s_j$ are \emph{paired} and write
                $s_i\sim s_j$.
            \item If $s_i\sim s_i$, then $b(s_i)\in\{\tred,\tblue\}$.
            \item If $s_i \sim s_j$ where $i<j$, then $b(s_i)\in\{\lred,\lblue\}$,
                $b(s_j)\in\{\rred,\rblue\}$, and their colours match.
        \end{itemize}
    \item \label{a3}
        \emph{Correct nesting:} If $s_i\sim s_j$ and $s_{i'}\sim s_{j'}$, then we cannot have
        $i<i'<j<j'$.
    \item \label{a4}
        \emph{No red/blue transition:} For every $i\in[n-1]$, we have $b(s_i)b(s_{i+1})\notin
        \{\rred\lblue,\rred\tblue,\tred\lblue,\tred\tblue\}$. 
\end{enumerate} 

The system \ref{a1}--\ref{a4} is an unsatisfiable set of constraints where each constraint depends on at
most $4$ indices of the string $s$. This means there are at most $O(n^4)$ constraints.

\subsection{CNF encodings}

\paragraph{Wide formula.}
Translating a system of~$\poly(n)$ constraints of~$O(1)$-arity over
strings~$s \in \Sigma^n$ to an~$O(\log | \Sigma |)$-CNF formula $F$ of
size~$\poly(n, \lvert \Sigma \rvert)$ is straightforward:
Introduce~$m\coloneqq \log | \Sigma |$ variables per
index~$i \in [n]$. Fix an arbitrary encoding function
$\pi\colon\Sigma\to\{0,1\}^m$ that encodes each symbol~$s_i$ as an
$m$-bit string. Then encode each $O(1)$-arity constraint
separately. It is immaterial here which precise encoding $\pi$ we use;
any choice will work. Note that because every constraint is of
constant arity it involves at most~$O(\log |\Sigma|)$ variables and is
thus translated to at most~$2^{O(\log |\Sigma|)} = \poly |\Sigma|$
clauses. Hence the size of the resulting formula $F$ is
indeed~$\poly(n,\lvert \Sigma \rvert)$.

\paragraph{Narrow formula.}
To be able to lift our tradeoff result from resolution to monotone
circuits we need a constant-width CNF. 
We obtain a constant width encoding of~$F$ by introducing extension
variables~$y_D$ for clauses~$D$ of constant \emph{index-width}: the
\emph{index-width}\footnote{Similar notions of width have previously
  been studied under the name pigeon-width or block-width.} of a
clause denotes the number of indices mentioned and the
\emph{index-width} of a formula is the maximum index-width of any
clause in it. Note that~$F$ has constant index-width since the
original constraints are of constant arity.

More precisely, if~$k$ denotes the index-width of~$F$, then we
introduce an extension variable~$y_D$ for every clause~$D$ of
index-width~$\leq k$ and encode~$F$ as a $3$-CNF~$F'$
of~$\poly(n, |\Sigma|)$ size over these variables: $F'$ consists of
\begin{itemize}
\item \emph{$F$-axioms}: if the clause~$D$ is in~$F$, then~$F'$
  contains the unary clause~$y_D$, and
\item \emph{extension axioms}:
  \begin{itemize}
  \item for every variable~$x$ of~$F$ the formula $F'$ contains
    the clauses~$y_x \lor y_{\neg x}$
    and~$\neg y_x \lor \neg y_{\neg x}$ ensuring that~$y_x$ is the
    negation of~$y_{\lnot x}$, and
  \item for clauses~$A,B,D$ of index-width~$\leq k$
    satisfying~$D = A \lor B$ the formula~$F'$ contains clauses
    enforcing the equivalence~$y_D \leftrightarrow (y_A \lor y_B)$.
  \end{itemize}
\end{itemize}
The index-width required to refute~$F$ is closely related to the width
needed to refute~$F'$. The proofs of the following propositions are
standard and provided in \cref{sec:proof-width-index-width}.

\begin{restatable}{proposition}{PropositionIndexWidth}
  \label{prop:index-width}
  If~$F$ admits an index-width-$w$ depth-$d$ resolution refutation,
  then~$F'$ admits a width-$O(w)$ depth-$O(dw)$ resolution refutation.
\end{restatable}

\begin{restatable}{proposition}{PropositionIndexWidthRev}
  \label{prop:index-width-rev}
  If~$F'$ admits a width-$w$ depth-$d$ resolution refutation, then~$F$
  admits an index-width-$O(w)$ depth-$O(d \log |\Sigma|)$ resolution
  refutation.
\end{restatable}


Since the constraints \ref{a1}--\ref{a4} are all of constant arity the
above translations apply. We denote the resulting $O(\log n)$-CNF
formula by~$\Bracket$ and the resulting $3$-CNF formula
by~$\ExBracket$. Note that by
\cref{prop:index-width,prop:index-width-rev} it suffices to prove
\cref{thm:width-ub,thm:depth-lb} for the more convenient
formula~$\Bracket$.

\section{Upper Bound}
\label{sec:upper-bound}
In this section, we prove the following resolution width upper bound for bracket formulas.
\TheoremUB*

It suffices to describe a refutation of $\Bracket$ of small index-width because of
\cref{prop:index-width}. That is, we think of the formula $\Bracket$ as an unsatisfiable system of
axioms over strings in~$\Sigma^n$. To describe the refutation, we use the standard top-down language of
\emph{prover--adversary games}~\cite{Pudlak00,Atserias08}.

\subsection{Prover--adversary games} \label{sec:pa-games}

The prover--adversary game for $\Bracket$ is played between two competing players, prover and adversary, and
proceeds in rounds. The state of the game in each round is modeled as a partial assignment $\rho \in
(\Sigma\cup\{*\})^n$. At the start of the game, $\rho\coloneqq *^n$. In each round:
\begin{itemize}[label = $-$, noitemsep]
    \item \emph{Query a symbol:} Prover specifies an index $i \in [n]$ and the adversary responds with a
        symbol $\sigma \in \Sigma$. The state $\rho$ is updated by $\rho_i \leftarrow \sigma$.
    \item \emph{Forget symbols:} Prover specifies a subset $I \subseteq [n]$ and we update
        $\rho_i \leftarrow *$ for all $i \in I$.
\end{itemize}
An important detail is that the adversary is allowed to choose $\sigma \in \Sigma$ afresh even if the
$i$-th symbol was queried and subsequently forgotten during past play. The game ends when $\rho$
falsifies some axiom of $F$. The prover's goal is to end the game while keeping the memory size of $\rho$
(number of non-$*$ entries) at most $w$. The least memory size $w$ so that prover has a winning strategy
characterizes the index-width of $\Bracket$~\cite{Pudlak00,Atserias08}.

\subsection{Prover strategy} \label{sec:prover-strategy}

During our prover strategy, we maintain the invariant that our state $\rho$ always contains whole bracket pairs. That is, whenever we query a symbol, we follow its pointer and query its mate. Similarly, if we forget a symbol, we also forget its mate. We tacitly assume that the adversary responds each query without falsifying the first three axioms \ref{a1}--\ref{a3} (otherwise we win immediately). Our goal as the prover then becomes to find the substring $\rred\lblue$.

At the heart of our prover strategy is a recursive procedure that takes as input a state $\rho$ that contains a red bracket pair~$\red R$ and a blue bracket pair $\blue B$ that lies to the right of $\red R$. The procedure outputs a state that has found either:
\begin{itemize}[label=$-$,noitemsep]
    \item a blue bracket pair enclosing $\red R$;
    \item a red bracket pair enclosing~$\blue B$; or
    \item the procedure outright discovers the substring $\rred\lblue$ that falsifies \ref{a4}.
\end{itemize}
That is, in pictures:
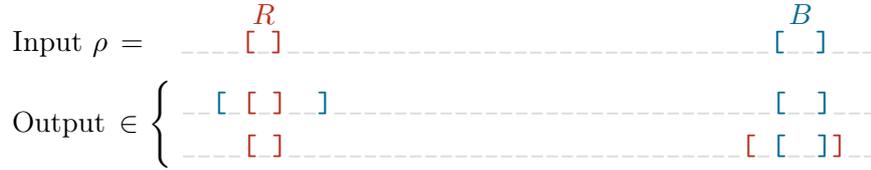
\vspace{-0.2em}
\begin{align*}
&
\hspace{3.75em} \red R
\hspace{17.66em} \blue B
\\[-.6em]
\text{Input}~\rho\,=&
\hspace{.93em}
~\texttt{\gray
\_\_\_\_{\red [}\_{\red ]}\_\_\_\_\_\_\_\_\_\_\_\_\_\_\_\_\_\_\_\_\_\_\_\_\_\_\_\_\_\_\_\_{\blue [}\_\_{\blue ]}\_\_\_
} \\[.3em]
\text{Output}\,\in\,&
\begin{cases}
~\texttt{\gray
\_\_{\blue [}\_{\red [}\_{\red ]}\_\_{\blue ]}\_\_\_\_\_\_\_\_\_\_\_\_\_\_\_\_\_\_\_\_\_\_\_\_\_\_\_\_\_{\blue [}\_\_{\blue ]}\_\_\_
} \\
~\texttt{\gray
\_\_\_\_{\red [}\_{\red ]}\_\_\_\_\_\_\_\_\_\_\_\_\_\_\_\_\_\_\_\_\_\_\_\_\_\_\_\_\_\_{\red [}\_{\blue [}\_\_{\blue ]}{\red ]}\_\_
}
\end{cases}
\end{align*}
In particular, if we invoke the procedure on the initial red/blue pair of brackets guaranteed by axiom \ref{a1}, the procedure must find the substring $\rred\lblue$. The procedure is:
\begin{enumerate}
\item If the pairs $\red R$ and $\blue B$ are touching ($\lred\cdots\rred\lblue\cdots\rblue$), then we have found our substring $\rred\lblue$.
\item Otherwise we query the middle index in the interval between $\red R$ and $\blue B$. Say the adversary responds with a blue bracket pair $\blue B'$. If $\blue B'$ encloses $\red R$, we are done.
\item Otherwise we recurse on $\red R$ and $\blue B'$.
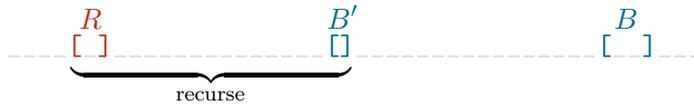
\begin{align*}
&
\hspace{3.75em} \red R
\hspace{7.75em} \blue B'
\hspace{8.8em} \blue B
\\[-.6em]
&
\hspace{.93em}
~\texttt{\gray
\_\_\_\_{\red [}\_{\red ]}\_\_\_\_\_\_\_\_\_\_\_\_\_\_{\blue []}\_\_\_\_\_\_\_\_\_\_\_\_\_\_\_\_{\blue [}\_\_{\blue ]}\_\_\_
}\\[-1.2em]
&\hspace{3.3em}\underbrace{\hspace{9.6em}}_{\text{recurse}}
\end{align*}
\item If this recursive call returns a blue bracket pair enclosing $\red R$, then we are done. Otherwise we have found a red middle pair $\red R'$ whose \emph{area} (number of indices enclosed by the pair) is larger than the area of $\blue B'$. We now forget $\blue B'$ and recurse on $\red R'$ and $\blue B$.
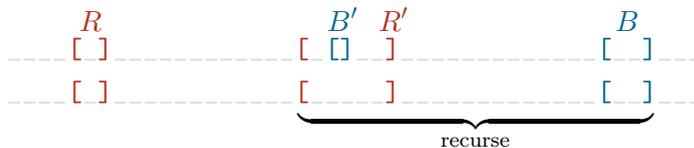
\begin{align*}
&
\hspace{3.75em} \red R
\hspace{7.75em} \blue B'
\hspace{0.7em} \red R'
\hspace{7.1em} \blue B
\\[-.6em]
&
\hspace{.93em}
~\texttt{\gray
\_\_\_\_{\red [}\_{\red ]}\_\_\_\_\_\_\_\_\_\_\_\_{\red [}\_{\blue []}\_\_{\red ]}\_\_\_\_\_\_\_\_\_\_\_\_\_{\blue [}\_\_{\blue ]}\_\_\_
}\\
&
\hspace{.93em}
~\texttt{\gray
\_\_\_\_{\red [}\_{\red ]}\_\_\_\_\_\_\_\_\_\_\_\_{\red [}\_\_\_\_\_{\red ]}\_\_\_\_\_\_\_\_\_\_\_\_\_{\blue [}\_\_{\blue ]}\_\_\_
}\\[-1.2em]
&\hspace{11.1em}\underbrace{\hspace{12.2em}}_{\text{recurse}}
\end{align*}
\item By continuing this way, we find a sequence of middle pairs whose area keeps increasing every step. Eventually, a recursive call must return a bracket enclosing $\red R$ or~$\blue B$ that is of the opposite colour. (It is possible that the middle pair, say $\red R'$, encloses the same-coloured input pair $\red R$ at some point. This only means that the next recursive call on $\red R'$ and $\blue B$ will return a desired oppositely-coloured output pair.) 
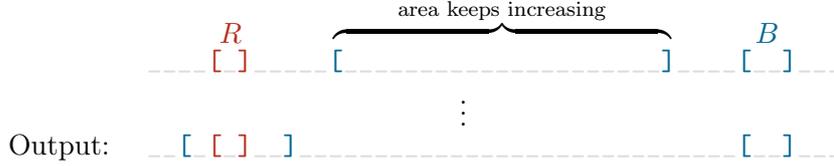
\vspace{-0.2em}
\begin{align*}
&
\hspace{7.5em} \overbrace{\hspace{11.6em}}^{\text{area keeps increasing}}
\\[-1.5em]
&
\hspace{3.75em} \red R
\hspace{17.66em} \blue B
\\[-.6em]
&
\hspace{.93em}
~\texttt{\gray
\_\_\_\_{\red [}\_{\red ]}\_\_\_\_\_{\blue [}\_\_\_\_\_\_\_\_\_\_\_\_\_\_\_\_\_\_\_\_\_{\blue ]}\_\_\_\_{\blue [}\_\_{\blue ]}\_\_\_
}\\[-.2em]
&\hspace{12em}\vdots
\\[-.5em]
\text{Output:}&
\hspace{.93em}
~\texttt{\gray
\_\_{\blue [}\_{\red [}\_{\red ]}\_\_{\blue ]}\_\_\_\_\_\_\_\_\_\_\_\_\_\_\_\_\_\_\_\_\_\_\_\_\_\_\_\_\_{\blue [}\_\_{\blue ]}\_\_\_
}
\end{align*}
\end{enumerate}
The recursive depth of this strategy is $O(\log n)$ and we keep at most $4$ bracket pairs in memory in each level of recursion. It follows that the total memory usage is $O(\log n)$ symbols. This shows that~$\Bracket$ has a refutation of index-width $O(\log n)$, which concludes the proof of \cref{thm:width-ub}.

\section{Lower Bound}
\label{sec:lower-bound}

In this section we prove the following depth lower bound for the
bracket formulas.

\TheoremLB*

Note that by \cref{prop:index-width-rev} it suffices to prove that any
index-width-$w$ resolution refutation of~$\Bracket$ has
depth~$n^{\Omega(\log n / \log w)}$. We prove this by exhibiting an
adversary strategy for the prover--adversary game such that if the
game state~$\rho \in (\Sigma \cup \set{*})^n$ is limited to memory
size~$w$, then the prover cannot win the prover--adversary game in
less than~$n^{\Omega(\log n / \log w)}$ rounds. Since the minimum
number of rounds before the prover can win corresponds to the minimum
depth of a resolution refutation of
width~$w$~\cite{Pudlak00,Atserias08} we may conclude
\cref{thm:depth-lb}.

The high-level objective of the adversary is to maintain the property
that differently coloured top-level brackets of~$\rho$ are far
apart. This is achieved by a strategy that forces the prover to
essentially follow the upper bound strategy of \cref{thm:width-ub}.
Throughout the game the adversary maintains a ``cleaned-up'' view of
the game state~$\rho$. We formalize this notion as a \emph{container}
of~$\rho$ in \cref{sec:container}. This section is followed by
\cref{sec:adversary-strategy} devoted to the adversary
strategy. Finally in \cref{sec:buffer,sec:proof-cover} we prove that
our adversary strategy indeed guarantees that the prover cannot win
the game in the first~$n^{\Omega(\log n / \log w)}$ rounds.

\subsection{Containers}
\label{sec:container}

The game state~$\rho \in (\Sigma \cup \set{*})^n$ is a somewhat unruly
partial assignment. It may contain brackets that are not paired or
bracket pairs that are redundant: if~$\rho$ assigns consecutive
indices to, say,~$\lred \tblue \rred$, then the blue trivial bracket
pair~$\tblue$ is of no use to the prover. In order to concisely state
the adversary strategy we want a ``cleaned up'' view of~$\rho$.

To this end we introduce the notion of a \emph{container}
of~$\rho$. Informally a container of~$\rho$ is a partial assignment
that (1) consists of bracket pairs and (2) is minimal while
guaranteeing that~$\rho$ does not falsify an axiom of the bracket
formula.
Thus, since a container ensures that~$\rho$ does not falsify an axiom,
the task of the adversary is reduced to maintaining a container
of~$\rho$ for the initial~$n^{\Omega(\log n / \log w)}$
rounds. In \cref{sec:adversary-strategy} we exhibit such a
strategy. In the remainder of this section we formalize the notion of
a container.

We say that a partial assignment is consistent if it falsifies no
axiom of the system \ref{a1}--\ref{a4}. Recall that for a
symbol~$\sigma \in \Sigma$ we write~$\sigma = (b(\sigma),p(\sigma))$
where~$b(\sigma) \in \set{\tred,\tblue,\lred, \lblue, \rred, \rblue}$
is the bracket type and~$p(\sigma) \in [n]$ is the pointer.

\begin{definition}[Configuration]
  A \emph{configuration~$\calC \in (\Sigma \cup \set{*})^n$} is a
  consistent partial assignment consisting of bracket pairs, that is,
  if~$\calC$ assigns index~$i$ and~$p(\calC_i) = j$, then~$\calC$ assigns
  index~$j$.
\end{definition}

For example the partial assignment~$\gray\lred\lblue\rblue\_\_\tblue$
is \emph{not} a configuration: it is consistent (assuming the pointers
are correctly chosen) but the bracket~$\lred$ is unpaired.

Note that a configuration~$\calC$ can be thought of as a set of
bracket pairs. In the following we thus interchangeably
identify~$\calC$ as a set of bracket pairs~$\calC = \set{P_i}_i$ and
as a consistent partial
assignment~$\calC \in (\Sigma \cup \set{*})^n$.

The \emph{area} of a bracket pair~$P$, denoted by $\area(P)$, is the
set of indices enclosed by~$P$ and we say that an index~$i\in[n]$
is \emph{covered} by~$P$ if~$i \in \area(P)$. These notions are
extended in the natural way to a configuration~$\calC$: the \emph{area
  of~$\calC$} consists of all indices covered by its bracket
pairs~$\area(\calC) = \bigcup_{P \in \calC} \area(P)$ and an index~$i$ is
\emph{covered} by~$\calC$ if~$i \in \area(\calC)$.

A configuration~$\calC$ is said to be \emph{closed} if it assigns all
covered indices and it is \emph{locally consistent} if there is a
closed configuration~$\calC'$ that extends~$\calC \subseteq \calC'$. For example,
the configuration \vspace{-0.2em}
\begin{align*}
  \texttt{
  \gray
  \_\_\_\lred\lblue\_\lred\_\_\_\_\rred\_\_\lblue\_\_\rblue\rblue\rred\_\_\_\_\_\_\_\_\_\_\_\_\_\_\_\_\_\_\_\_\_\lblue\_\_\rblue\_\_
  }
  \hspace{.75em}
\end{align*}
is locally consistent as witnessed by the closed configuration
\vspace{-0.2em}
\begin{align}
  \label{eq:closed-config}
  \texttt{
  \gray
  \_\_\_\lred\lblue\lblue\lred\lred\lred\rred\rred\rred\rblue\hspace{-0.16em}\tblue\hspace{-0.16em}\lblue\lblue\rblue\rblue\rblue\rred\_\_\_\_\_\_\_\_\_\_\_\_\_\_\_\_\_\_\_\_\_\lblue\lblue\rblue\rblue\_\_
  } \eqperiod
\end{align}
By contrast, the following configuration is \emph{not} locally consistent:
\vspace{-0.2em}
\begin{align*}
  \texttt{
  \gray
  \_\_\_\lred\lblue\lred\_\_\_\_\_\rred\_\_\lblue\_\_\rblue\rblue\rred\_\_\_\_\_\_\_\_\_\_\_\_\_\_\_\_\_\_\_\_\_\lblue\_\_\rblue\_\_
  } \eqperiod
\end{align*}
(There is no consistent assignment for the two indices to the right of the left-most~$\rred$.)

The configuration~$\calC$ maintained by the adversary is a locally
consistent configuration that contains \emph{top-level bracket pairs}
only: a bracket pair~$P \in \calC$ is top-level if for all other
pairs~$P' \in \calC$ it holds that $\area(P) \not\subseteq \area(P')$. For
example, the configuration
\vspace{-0.2em}
\begin{align*}
  \texttt{
  \gray
  \_\lred\_\lblue\lred\_\rred\_\rblue\_\_\_\_\_\_\rred\_\lred\_\rred\_\lred\_\_\rred\_\_\_\lblue\_\lred\_\rred\_\rblue\_\_\lblue\_\_\_\_\_\_\_\rblue\_
  }
  \hspace{.75em}
\end{align*}
contains the top-level bracket pairs
\vspace{-0.2em}
\begin{align*}
  \texttt{
  \gray
  \_\lred\_\_\_\_\_\_\_\_\_\_\_\_\_\rred\_\lred\_\rred\_\lred\_\_\rred\_\_\_\lblue\_\_\_\_\_\rblue\_\_\lblue\_\_\_\_\_\_\_\rblue\_
  } \eqperiod
\end{align*}

A configuration~$\calC$ is \emph{monotone} if the red top-level bracket
pairs of~$\calC$ are before some index~$i$ while the blue top-level
bracket pairs come after~$i$. The \emph{separating interval} of such a
monotone configuration~$\calC$ is the (unique) maximal
interval~$I \subseteq [n]$ that contains~$i \in I$ and satisfies that
no index in~$I$ is covered by~$\calC$. The above example configuration
is monotone and has a separation interval of size~$3$. 

The locally consistent configuration~$\calC$ maintained by the adversary
is furthermore monotone and related to the game
state~$\rho \in (\Sigma \cup \set{*})^n$ as follows. Let the
\emph{support} of a partial assignment~$\tau$, denoted
by~$\supp(\tau)$, be the set of indices assigned by~$\tau$, that is,
$\supp(\tau) = \set{i \in [n]: \tau_i \neq *}$.

\begin{definition}[Domination]
  \label{def:domination}
  A configuration~$\calC$ \emph{dominates} a partial assignment~$\tau$ if
  the support of~$\tau$ is covered by~$\calC$
  (\ie $\supp(\tau) \subseteq \area(\calC)$) and there is a closed
  configuration~$\calC'$ that extends~$\calC\subseteq \calC'$ as well
  as~$\tau\subseteq \calC'$.
\end{definition}

For example the partial assignment
\vspace{-0.2em}
\begin{align*}
  \texttt{
  \gray
  \_\_\_\_\_\_\lred\_\_\_\_\_\_\_\_\_\rblue\_\_\_\_\_\_\_\_\_\_\_\_\_\_\_\_\_\_\_\_\_\_\_\_\lblue\_\_\_\_\_
  }
  \hspace{.75em}
\end{align*}
is dominated, assuming the pointers are appropriately chosen, by
\vspace{-0.2em}
\begin{align*}
  \texttt{
  \gray
  \_\_\_\lred\_\_\_\_\_\_\_\_\_\_\_\_\_\_\_\rred\_\_\_\_\_\_\_\_\_\_\_\_\_\_\_\_\_\_\_\_\_\lblue\_\_\rblue\_\_
  }
  \hspace{.75em}
\end{align*}
as witnessed by \eqref{eq:closed-config}.

To summarize the above discussion the adversary maintains a monotone
and locally consistent configuration~$\calC$ that dominates the game
state~$\rho$. Furthermore, the configuration~$\calC$ consists of top-level
bracket pairs only. We impose in fact a slightly stronger condition
than only containing top-level pairs as summarized in the next
definition.

\begin{definition}[Container]
  \label{def:container}
  A \emph{container} of a partial assignment~$\tau$ is a locally
  consistent and monotone configuration~$\calC$ such that
  \begin{enumerate}[label=(C\arabic*),noitemsep]
  \item $\calC$ dominates~$\tau$, and
    \label[condition]{def:container-domination}
  \item $\calC$ is minimal while dominating~$\tau$: removing any bracket
    pair from~$\calC$ results in a configuration that does not
    dominate~$\tau$.
    \label[condition]{def:container-minimal}
  \end{enumerate}
\end{definition}

Observe that \ref{def:container-minimal} implies that containers consist of top-level bracket pairs
only. It further allows us to relate the size
of a container to the size of the support of~$\tau$ as follows.

\begin{proposition}
  \label{prop:domination-local-consistency}
  \label{prop:container-size}
  Consider a partial assignment~$\tau$ and a configuration~$\calC$.
  If~$\calC$ is a container of~$\tau$, then~$\tau$ is consistent and,
  furthermore,~$\calC$ contains at most~$\lvert\supp(\tau)\rvert$ many
  bracket pairs.
\end{proposition}
\begin{proof}
  By \ref{def:container-domination} we have that~$\calC$
  dominates~$\tau$ and hence by \cref{def:domination} there is a
  configuration~$\calC'$ that extends~$\tau \subseteq \calC'$. Since
  configurations are consistent this implies that~$\tau$
  is consistent. The claim about the number of bracket pairs in~$\calC$
  follows from the minimality property \ref{def:container-minimal}.
\end{proof}

By \cref{prop:domination-local-consistency} it holds that while the
adversary maintains a container of the game state~$\rho$ the prover
cannot win the prover--adversary game. The following section exhibits
an adversary that maintains a container of~$\rho$ for the
initial~$n^{\Omega(\log n / \log w)}$ rounds.

\subsection{The Adversary Strategy}
\label{sec:adversary-strategy}

This section is devoted to the adversary strategy. Throughout we
assume that the memory size of the game state~$\rho$ is bounded
by~$w$.

The main objective of the adversary is to maintain a container~$\calC$
of~$\rho$ for the initial~$n^{\Omega(\log n/\log w)}$ rounds
of the prover--adversary game. By
\cref{prop:domination-local-consistency} it thus follows that the
prover cannot win the game in the
first~$n^{\Omega(\log n/\log w)}$ rounds. \Cref{thm:depth-lb}
follows from the fact that the minimum number of rounds required for
the prover to win the game corresponds to the minimum depth of a
width-$w$ resolution refutation.

In addition to the container~$\calC$ of~$\rho$ the adversary also maintains
an interval~$I\subseteq [n]$. Initially~$I = [n]$ and we should think
of~$I$ as the separation interval of~$\calC$. Let us stress, though, that
for technical convenience the actual strategy only guarantees that~$I$
is a sub-interval of the separation interval.

Let us describe the adversary strategy. There are two integer
parameters~$\ell_0 = \floor{\eps \frac{\log n}{\log w}}$
and~$d = \lfloor n^\eps
\rfloor$. 
The adversary repeats the following~$d$ times.
\begin{enumerate}
\item Identify a large interval~$I$ with no index covered by the
  container~$\calC$.
  \label[step]{step:identify}
\item Move the separation interval to~$I$ by modifying~$\calC$.
  \label[step]{step:move}
\item Recursively invoke the adversary on the interval~$I$.
\end{enumerate}
Once the adversary has reached recursion level~$\ell_0$ they are ready
to play a round of the prover--adversary game. Suppose that the prover
queries some index~$i$. If index~$i$ is covered by the
container~$\calC$, then the adversary answers according to the complete
configuration~$\calC' \supseteq \calC \cup \rho$ as guaranteed to exist by
\ref{def:container-domination}. Otherwise, if~$i$ is not covered by~$\calC$,
for~$[a,b] = I$ the adversary answers with a trivial bracket coloured
red if~$i < \frac{a+b}{2}$ and coloured blue otherwise.

This essentially completes the description of the adversary strategy
modulo the ``move'' operation performed in \cref{step:move}. This is
best explained by a picture. Say the current container is
\vspace{-0.2em}
\begin{align*}
  \calC&=
  \texttt{
  \gray
  \_\lred\_\_\_\_\_\rred\_\_\_\_\_\_\_\lred\_\rred\_\lred\_\_\_\rred\_\_\_\_\lred\_\_\_\_\_\rred\_\_\lblue\_\_\_\_\_\_\_\rblue\_
     } \eqcomma
  \\[-1.2em]
&\hspace{5.8em}\underbrace{\hspace{3.75em}}_{I}
\end{align*}
where~$I$ is the interval as identified in \cref{step:identify}. The
move operation replaces any red (blue) bracket pair to the right (to
the left) of~$I$ by a blue (red) bracket pair that minimally covers
it. If any of these newly added pairs overlap, then they are merged
into a single bracket pair.

For the above example container~$\calC$ the move operation results in the new
container
\vspace{-0.2em}
\begin{align*}
  \calC'&=
  \texttt{
  \gray
  \_\lred\_\_\_\_\_\rred\_\_\_\_\_\_\lblue\_\_\_\_\_\_\_\_\_\rblue\_\_\lblue\_\_\_\_\_\_\_\rblue\_\lblue\_\_\_\_\_\_\_\rblue\_
  } \eqperiod
\end{align*}
Observe that if~$\calC$ is a container of~$\rho$, then by construction~$\calC'$
is also a container of~$\rho$. Further the area of~$\calC'$ is quite
closely related to the area of~$\calC$: according
\cref{prop:container-size} any container of~$\rho$ contains at most~$w$
bracket pairs and hence~$\calC'$ may only cover an additional~$2w$
indices. Finally note that the~$\textsc{Move}$ operation is the only
source of non-trivial bracket pairs in the container.

This completes the description of the adversary strategy. For a more
thorough treatment we refer to \cref{alg:adversary}. The strategy is
initially invoked by~$\textsc{Adversary}(\emptyset,[n],0)$.

\begin{algorithm}[p]
  \caption{The adversary strategy.}
  \label{alg:adversary}
  \begin{algorithmic}[1]
    \Procedure{Adversary}{$\calC, I,\ell$}
    \If{$\lvert I \rvert < n/(4w)^{\ell}$}
    \State \textbf{return} {\red failure}
    \EndIf
    \Statex
    \If{$\ell = \ell_0$}
    \State $\tilde \calC_d \gets$ \Call{PlayRound}{$\calC, I$}  
    \Else
    \State $\tilde \calC_0 \gets \calC$
    \For{$i= 0, 1, \ldots, d - 1$}
      \State $\tilde I_{\indexBeforeCall} \gets$ a largest interval in~$I$ with no
      index covered by~$\tilde \calC_{\indexBeforeCall}$
      \State $\calC_{\indexBeforeCall} \gets $\Call{Move}{$\tilde \calC_{\indexBeforeCall}, \tilde I_{\indexBeforeCall}$}
      \State $I_{\indexBeforeCall} \gets $ maximum interval in~$\tilde I_{\indexBeforeCall}$ with no
      index covered by~$\calC_{\indexBeforeCall}$
      \State $\tilde \calC_{\indexAfterCall} \gets $\Call{Adversary}{$\calC_{\indexBeforeCall}, I_{\indexBeforeCall}, \ell+1$}
    \EndFor
    \EndIf
    \State \textbf{return} $\tilde \calC_d$
    \EndProcedure
  \end{algorithmic}
\end{algorithm}

\begin{algorithm}[p]
  \caption{Plays a round of the prover--adversary game.}
  \label{alg:play}
  \begin{algorithmic}[1]
    \Procedure{PlayRound}{$\calC, I$}
    \If{prover forgets}
    \State $\calC \gets$ minimum sub-configuration of~$\calC$ that dominates~$\rho$
    \Else
    \State $i \gets$ index queried by the prover
    \If{$i \in \area(\calC)$}
    \State $\calC' \gets$ closed configuration~$\calC' \supseteq \calC \cup \rho$ as
    guaranteed by \cref{def:domination}
    \State $\rho_i \gets \calC'_i$
    \Else
    \State $[a,b] \gets I$
    \If{$i < \frac{a+b}{2}$}
    \State $\rho_i \gets \tred$; $\calC_i \gets \tred$
    \Else
    \State $\rho_i \gets \tblue$; $\calC_i \gets \tblue$
    \EndIf
    \EndIf
    \EndIf
    \State \textbf{return} $\calC$
    \EndProcedure
  \end{algorithmic}
\end{algorithm}

\begin{algorithm}[p]
  \caption{Returns a configuration~$\calC'$ such that the separation
    interval intersects~$I$.}
  \label{alg:move}
  \begin{algorithmic}[1]
    \Procedure{Move}{$\calC, I$}
    \State $\calP \gets$ red (blue) bracket pairs to the right (left)
    of~$I$
    \State $\calP' \gets \emptyset$
    \For{$P \in \calP$}
    \State $P' \gets$ minimum bracket pair that strictly contains~$P$ of opposite colour
    \State $\calP' \gets \calP' \cup \set{P'}$
    \EndFor
    \While{$\exists P,P'\in\calP'$ such that $\area(P) \cap \area(P') \neq
      \emptyset$}
    \State $P'' \gets$ bracket pair such that~$\area(P'') = \area(P)
    \cup \area(P')$ of equal colour
    \State $\calP' \gets \calP'$ with $P,P'$ replaced by $P''$
    \EndWhile
    \State \textbf{return} $(\calC \setminus \calP) \cup \calP'$
    \EndProcedure
  \end{algorithmic}
\end{algorithm}

Let us remark that the game--state~$\rho$ is treated rather
implicitly; we do not keep explicit track of~$\rho$. It is understood,
however, that we are playing a prover--adversary game and in each
round played~$\rho$ changes. This implies in particular that the game
state when an adversary returns is (most likely) distinct from the
game state when the adversary was instantiated.

Finally observe that the adversary changes the maintained~$\calC$ in
two places only: in \textsc{PlayRound} when playing a round of the
game and in \textsc{Move} when moving the separation
interval. Hence~$\calC$ may fail to be a container of~$\rho$ in only
these two places. The following claim shows that if the adversary does
not fail, then~$\calC$ is indeed a container of~$\rho$.

\begin{lemma}
  \label{lem:fail}
  Unless the adversary {\red fails} it holds that the maintained
  configuration~$\calC$ is a container of the game state.
\end{lemma}

\begin{proof}
  Consider an adversary invoked
  by~$\textsc{Adversary}(\calC,I,\ell)$. By construction it should be
  evident that the interval~$I$ is a sub-interval of the separation
  interval of~$\calC$. Since the adversary does not fail we may assume
  that~$\lvert I \rvert \geq n/(4w)^{\ell_0} \geq 10$ by our choice
  of~$\ell_0 = \lfloor \eps \frac{\log n}{\log w}\rfloor$.

  By induction on the
  recursion\nobreakdash-depth~$\ell = \ell_0, \ldots, 0$ we prove that
  if~$\calC'=\textsc{Adversary}(\calC,I,\ell)$ is instantiated with a
  container~$\calC$ of the current game state, then the adversary returns
  a container~$\calC'$ of the resulting game state with a separation
  interval of size at least~$5$. Note that once the induction is
  established the statement follows for the initial
  call~$\textsc{Adversary}(\emptyset,[n], 0)$ since the empty
  configuration is a container for the initial game
  state~$\rho = *^n$.

  The base case, that is for an adversary of recursion-depth~$\ell_0$,
  the inductive hypothesis holds by inspection of the procedure
  $\textsc{PlayRound}$ and the assumption that the interval~$I$ is of
  size at least~$10$. Note that the separation interval of the
  container returned is of size at least~$5$.

  We may thus assume the inductive hypothesis for
  depth~$\ell+1$. Consider an adversary of depth~$\ell$. We need to
  argue that \textsc{Move} does not cause the container to become
  inconsistent. This is readily verified: note that the bracket
  pairs~$\calP$ being replaced in
  $\textsc{Move}(\tilde \calC_\indexBeforeCall, \tilde
  I_\indexBeforeCall)$ lie in-between the
  intervals~$\tilde I_\indexBeforeCall$ and the separation interval
  of~$\tilde \calC_\indexBeforeCall$. By induction we may assume that
  the separation interval of~$\tilde \calC_\indexBeforeCall$ is of
  size at least~$5$ and by assumption it holds
  that~$\lvert \tilde I_\indexBeforeCall\rvert \geq 10$ as otherwise
  the adversary would fail in the coming recursion. This implies in
  particular that the bracket pairs~$P' \in \calP'$ do not intersect
  the area of the configuration~$\calC \setminus \calP$. This
  establishes the induction. The statement follows.
\end{proof}

\subsection{Buffers}
\label{sec:buffer}

It remains to establish that every
adversary~$\textsc{Adversary}(\calC, I, \ell)$ is invoked with an
interval~$I$ of size~$\lvert I \rvert \geq n/(4w)^\ell$.
Let~$A_\ell = w(A_{\ell+1} + 3d)$ with~$A_{\ell_0} = 3d$
for~$\ell\in \set{0, 1, \ldots, \ell_0}$.

\begin{lemma}[Cover Lemma]
  \label{lem:cover}
  Consider the adversary~$\textsc{Adversary}(\calC, I, \ell)$ of
  recursion depth~$\ell$, suppose
  that~$\lvert I \rvert \geq n/(4w)^\ell$, and
  let~$\tilde \calC_\indexAfterCall =
  \textsc{Adversary}(\calC_{\indexBeforeCall}, I_{\indexBeforeCall},
  \ell+1)$ as in \cref{alg:adversary}. It holds
  that~$\lvert \area(\tilde \calC_i) \cap I\rvert \leq w(A_{\ell+1} +
  3i)$.
\end{lemma}

Let us first argue that \cref{lem:cover} is indeed sufficient, that
is, if \cref{lem:cover} holds, then the intervals with which the
adversaries are invoked are large.

\begin{lemma}[Size Lemma]
  \label{lem:large}
  Consider the adversary~$\textsc{Adversary}(\calC, I, \ell)$ of
  recursion depth~$\ell$, suppose
  that~$\lvert I \rvert \geq n/(4w)^\ell$, and
  let~$\tilde \calC_\indexAfterCall =
  \textsc{Adversary}(\calC_{\indexBeforeCall}, I_{\indexBeforeCall},
  \ell+1)$ as in \cref{alg:adversary}.  If it holds
  that~$\lvert \area(\tilde \calC_\indexBeforeCall) \cap I\rvert \leq
  A_\ell$,
  then~$\lvert I_{\indexBeforeCall} \rvert \geq n/(4w)^{\ell+1}$.
\end{lemma}

\begin{proof}
  We need to argue that there is a large
  interval~$\tilde I_\indexBeforeCall \subseteq I$ such
  that~$\tilde \calC_\indexBeforeCall$ covers no index
  in~$\tilde I_\indexBeforeCall$.
  By minimality of the configuration~$\calC$ (formally this can be
  established by a straightforward induction) it holds that~$\calC$
  contains at most~$w$ top-level bracket pairs. Hence there is such an
  interval~$\tilde I_\indexBeforeCall$ of size
  \begin{align}
    \lvert
    \tilde I_\indexBeforeCall
    \rvert
    \geq
    \frac{
    \lvert
    I
    \setminus
    \area
    \bigl(
    \tilde \calC_\indexBeforeCall
    \bigr)
    \rvert
    }{
    w+1
    } \eqperiod
    \label{eq:I}
  \end{align}
  By assumption and our choice of parameters it holds
  that~$\lvert I \rvert \geq n/(4w)^{\ell} \geq
  \Omega(n^{1-2\eps})$ while
  \begin{align}
    \lvert
    I
    \cap
    \area
    \bigl(
    \tilde \calC_\indexBeforeCall
    \bigr)
    \rvert
    \leq
    A_\ell
    \leq
    O(d\ell_0w^{\ell_0})
    =
    O(n^{3\eps})
    \ll
    \lvert I \rvert/2 \eqperiod
  \end{align}
  Substituting the above estimates into \cref{eq:I} reveals that
  \begin{align}
    \lvert
    \tilde I_\indexBeforeCall
    \rvert
    \geq
    \frac{n}
    {(4w)^\ell2(w+1)}
    \eqperiod
  \end{align}
  This implies that any invocation
  of~$\textsc{Adversary}(\calC_\indexBeforeCall,I_\indexBeforeCall,\ell+1)$
  at recursion-depth~$\ell+1$ ends up with an
  interval~$I_\indexBeforeCall$ of size at
  least~$\lvert I_\indexBeforeCall \rvert \geq n/((4w)^\ell2(w+1)) - 1
  > n/(4w)^{\ell+1}$. Note the minus one due to the replacement of
  bracket pairs in \textsc{Move}. This establishes the statement.
\end{proof}

With the \nameref{lem:large} at hand it remains to prove the
\nameref{lem:cover} in order to conclude \cref{thm:depth-lb}.

We prove the \nameref{lem:cover} by induction on the
recursion--depth~$\ell = \ell_0, \ldots, 0$ and~$i=0, 1, \ldots,
d$. Let us give a naïve proof sketch to then explain how to improve
the bound. The formal proof of the \nameref{lem:cover} is provided in
\cref{sec:proof-cover}.

Suppose by induction
that~$\tilde \calC_\indexAfterCall =
\textsc{Adversary}(\calC_{\indexBeforeCall}, I_{\indexBeforeCall},
\ell+1)$ covers at most~$A_{\ell+1}$ indices
in~$I_{\indexBeforeCall}$. If we apply the inductive hypothesis
naïvely we obtain that~$\tilde \calC_\indexAfterCall$ covers at
most~$(\indexAfterCall)(A_{\ell+1} + 3w)$ indices in~$I$. This is
clearly insufficient since~$\indexAfterCall \in [d]$
and~$d = \lfloor n^\eps \rfloor \gg w$.

The above argument has not used the fact that the game state~$\rho$ is
limited to memory size~$w$. We intend to argue that the game state may
contain at most~$w$ bracket pairs originating from distinct recursive
calls. This allows us to then conclude that~$\tilde \calC_d$ covers at
most~$w(A_{\ell+1} + 3d) = A_\ell$ indices in~$I$ as claimed.

To this end we show that non-trivial bracket pairs obtained from
different recursive calls are far apart. In a bit more detail,
for~$i \neq j$, we intend to argue that two non-trivial bracket
pairs~$P_i \in \tilde \calC_i \setminus \calC_{i-1}$
and~$P_j \in \tilde \calC_j \setminus \calC_{j-1}$ are at distance at
least~$4dw$. Since the \textsc{Move} operation increases the cover by
at most~$2w$ the large gap between~$P_i$ and~$P_j$ cannot be
bridged. Hence if in~$\tilde \calC_d$ both bracket pairs are
remembered, then these are distinct bracket pairs.

Making the above intuition into a formal proof requires some care. To
argue that non-trivial bracket pairs from distinct recursive calls are
far apart we introduce the notion of a \emph{buffer}.

\begin{definition}[Buffer]
  For integer~$s \in \N^+$ and an interval~$I = [a,b] \subseteq [n]$ we
  say that a configuration~$\calC$ has an~\emph{$s$-buffer in~$I$} if
  every bracket pair~$P \in \calC$ satisfies that if~$P$ is contained in
  the interval~$[1,a+s]$, then it is a trivial red pair~$P = \tred$
  and if~$P$ is contained in the interval~$[b-s,n]$, then it is a
  trivial blue pair~$P = \tblue$.
\end{definition}

The following lemma establishes that bracket pairs originating from
distinct recursive calls are far apart. In terms of buffers we
establish that bracket pairs
in~$\tilde \calC_\indexAfterCall \setminus \calC_{\indexBeforeCall}$
have a buffer in~$I_{\indexBeforeCall}$. The lemma is stated in a way
so that it can be readily applied in the inductive proof of the
\nameref{lem:cover}.

\begin{lemma}[Buffer Lemma]
  \label{lem:buffer}
  Consider the adversary~$\textsc{Adversary}(\calC, I, \ell)$ of
  recursion depth~$\ell$, fix~$c_\ell = 4dw - (\ell_0 - \ell)$, and
  let~$\tilde \calC_\indexAfterCall =
  \textsc{Adversary}(\calC_{\indexBeforeCall}, I_{\indexBeforeCall},
  \ell+1)$ as in \cref{alg:adversary}.
  If for some~$i\in[d-1]$ it holds that the
  configuration~$\tilde \calC_{i} \setminus \calC$ has a
  $c_\ell$-buffer in~$I$, the
  configuration~$\tilde \calC_{i} \setminus \calC_{i-1}$ has a
  $c_{\ell+1}$-buffer in~$I_{i-1}$, and the
  configuration~$\tilde \calC_{i+1} \setminus \calC_i$ has a
  $c_{\ell+1}$-buffer in~$I_{i}$,
  then~$\tilde \calC_{i+1} \setminus \calC$ has a~$c_\ell$-buffer
  in~$I$.
\end{lemma}

\begin{proof}
  Let~$[a,b] = I$ and consider a bracket
  pair~$P \in \tilde \calC_{i+1} \setminus \calC$.
  If~$P \in \tilde \calC_{i+1} \setminus \calC_i$
  a~$c_{\ell+1}$-buffer in~$I_i$, then there is nothing to show
  since~$I_i \subseteq I$ and
  hence~$\tilde \calC_{i+1} \setminus \calC_i$ is also
  a~$c_\ell$-buffer in~$I$.

  Otherwise~$P \in \calC_i \setminus \calC$.
  If~$P \in \tilde \calC_{i} \setminus \calC$, then there is nothing
  to show since~$\tilde \calC_{i} \setminus \calC$ is
  a~$c_\ell$-buffer in~$I$.
  We are left to show the statement for bracket
  pairs~$P \in \calC_{i} \setminus \tilde\calC_i$. Note that these
  bracket pairs~$P$ were added by~\textsc{Move} and are thus
  non-trivial. We need to establish that all such pairs~$P$ are
  contained in the interval~$[a+c_\ell+1,b-(c_\ell+1)]$.

  Observe that all bracket pairs
  in~$\calC_{i} \setminus \tilde \calC_{i}$ are of equal colour since
  the maintained configurations are monotone. Without loss of
  generality we may assume that the replaced bracket
  pairs~$\tilde \calC_{i} \setminus \calC_{i}$ are {\blue blue} and
  that the newly added pairs~$\calC_{i} \setminus \tilde \calC_{i}$
  are {\red red}.

  Denote by~$a_{i-1}$ the left endpoint of the interval~$I_{i-1}$.
  Since the configuration~$\tilde \calC_{i} \setminus \calC_{i-1}$ has
  a $c_{\ell+1}$-buffer in~$I_{i-1}$ it holds that the~$c_{\ell+1}$
  left-most indices of~$I_{i-1}$ do \emph{not} contain a blue bracket
  pair. Further using the fact that~$\calC_{i-1}$ covers no positions
  in~$I_{i-1}$ we obtain that all bracket
  pairs~$\tilde \calC_{i} \setminus \calC_{i}$ being replaced are
  contained in the interval~$[a_{i-1} + c_{\ell+1} + 1, b]$ and the
  new bracket pairs are thus contained in the
  interval~$[a_{i-1} + c_{\ell} + 1, b]$.
  Furthermore, since~$I_{i}$ is
  large~$\lvert I_{i} \rvert \gg c_{\ell+1}$ for our choice of
  parameters it holds that the newly added bracket
  pairs~$\calC_{i} \setminus \tilde \calC_{i}$ cannot cover
  the~$c_\ell$ right-most indices of~$I$. Hence all these bracket
  pairs are contained in the interval~$[a+c_\ell+1,b-(c_\ell+1)]$ as
  required. This establishes the statement.
\end{proof}

As previously explained the notion of a buffer lets us argue that
bracket pairs from different recursive calls are far apart. The next
definition formalizes what it means for a bracket pair to originate
from a recursive call.

\begin{definition}[Legacy, ancestors]
  Consider~$\textsc{Adversary}(\calC, I, \ell)$ of recursion
  depth~$\ell$ and
  let~$\tilde \calC_\indexAfterCall =
  \textsc{Adversary}(\calC_{\indexBeforeCall}, I_{\indexBeforeCall},
  \ell+1)$ as in \cref{alg:adversary}.
  Consider a non-trivial bracket pair~$P \in \tilde \calC_{i}$. If~$P$
  is not present in~$\calC_{i-1}$, then~$P$ is a \emph{legacy-$i$
    bracket pair}.  If~$\tilde \calC_{i-1}|_{P}$ denotes the set of
  non-trivial bracket pairs~$P' \in \tilde \calC_{i-1}$
  satisfying~$\area(P') \subseteq \area(P)$, then the \emph{legacy
    ancestors} of a bracket pair~$P$ are
  \begin{align*}
    \ancestors(P)
    &=
      \begin{cases}
        \set{P}
        &\text{if~$P$ is a legacy-$i$ bracket pair, and}\\
        \bigcup_{
        P' \in \tilde \calC_{i-1}|_{P}
        }
        \ancestors(P')
        &\text{otherwise.}
      \end{cases}
  \end{align*}
\end{definition}

The area of a bracket pair is contained in the area of its legacy
ancestors except for a few indices as stated next.

\begin{lemma}
  \label{lem:area-decompose}
  Consider the adversary~$\textsc{Adversary}(\calC, I, \ell)$ of
  recursion depth~$\ell$ and
  let~$\tilde \calC_\indexAfterCall =
  \textsc{Adversary}(\calC_{\indexBeforeCall}, I_{\indexBeforeCall},
  \ell+1)$ as in \cref{alg:adversary}.
  For all~$i \in \set{0,1, \ldots, d}$ there is a
  subset~$U \subseteq [n]$ of size~$\lvert U \rvert \leq 3iw$ such
  that for any bracket pair~$P \in \tilde \calC_{i}$ contained in
  the interval~$\area(P) \subseteq I$ it holds that
  \begin{align*}
    \area(P) \subseteq U \cup
    \bigcup_{P' \in \ancestors(P)} \area(P')\eqperiod
  \end{align*}
\end{lemma}
\begin{proof}
  By induction on~$i = 0, 1, \ldots, d$. For~$i = 0$ there is nothing
  to show since~$\tilde \calC_0 = \calC$ and~$\calC$ covers no index
  in~$I$ by construction.

  We may thus assume the statement for some~$i$ to prove the statement
  for~$i+1$. Recall that
  \begin{align}
    \lvert \area(\calC_{i}) \setminus \area(\tilde \calC_{i})\rvert \leq 2w
  \end{align}
  since the \textsc{Move} operation increases the area of the
  resulting configuration by at most~$2w$. Add all the indices in the
  difference to~$U_{i+1}$ along with all trivial bracket pairs
  in~$\tilde \calC_{i}$. Observe that~$\lvert U_{i+1} \rvert \leq 3w$
  and that it holds that
  \begin{align}
    \area(P)
    \subseteq
    U_{i+1}
    \cup
    \bigcup_{P' \in \tilde \calC_{i}|_{P}}
    \area(P')
    \eqperiod
  \end{align}
  The statement follows by appealing to the inductive hypothesis for
  each~$P' \in \tilde \calC_i|_P$ separately.
\end{proof}

Finally we are in a position to formally argue that a bracket pair
originates from a single legacy, that is, from a single recursive
call.

\begin{lemma}
  \label{lem:single-legacy}
  Consider the adversary~$\textsc{Adversary}(\calC, I, \ell)$ of
  recursion depth~$\ell$ and
  let~$\tilde \calC_\indexAfterCall =
  \textsc{Adversary}(\calC_{\indexBeforeCall}, I_{\indexBeforeCall},
  \ell+1)$ as in \cref{alg:adversary}.
  If for all~$j \leq i$ it holds
  that~$\tilde \calC_{j} \setminus \calC_{j-1}$ has
  a~$c_{\ell+1}$-buffer in~$I_{j-1}$, then for every bracket
  pair~$P \in \tilde \calC_{i}$ there is a~$j \leq i$ such that all
  legacy-ancestors of~$P$ are legacy-$j$ bracket pairs.
\end{lemma}

\begin{proof}
  By contradiction. Suppose there are two bracket
  pairs~$P_1,P_2 \in \ancestors(P)$ from distinct
  legacies~$j_1 < j_2$. Without loss of generality we may assume that
  there is no bracket pair~$P^\star \in \ancestors(P)$ that lies
  in-between~$P_1$ and~$P_2$.

  By definition of an ancestor there is a non-trivial bracket
  pair~$P'_1 \in \tilde \calC_{j_2}$ such
  that~$\area(P'_1) \supseteq \area(P_1)$. By assumption
  $\tilde \calC_{j_2} \setminus \calC_{j_2-1}$ has
  a~$c_{\ell+1}$-buffer in~$I_{j-1}$. Hence the distance between the
  bracket pairs~$P'_1$ and~$P_2$ is strictly larger than~$3dw$.
  
  Since there are no other legacy-ancestors in-between~$P'_1$
  and~$P_2$ by \cref{lem:area-decompose} all these indices
  between~$P'_1$ and~$P_2$ need to be covered by the set~$U$ of
  size~$\lvert U \rvert \leq 3dw$. This cannot be; the statement
  follows.
\end{proof}

\subsection{Proof of the~\nameref{lem:cover}}
\label{sec:proof-cover}

As previously explained we prove the \nameref{lem:cover} by induction
on the recursion--depth~$\ell = \ell_0, \ldots, 0$ and
on~$i=0,1, \ldots, d$. It is most convenient to think of the recursion
as an outer recursion on~$\ell$ and an inner recursion on~$i$.

The outer induction on the recursion--depth~$\ell = \ell_0, \ldots, 0$
establishes the following statement.
For~$\calC' = \textsc{Adversary}(\calC, I, \ell)$ of recursion
depth~$\ell$ it holds that if~$\lvert I \rvert \geq n/(4w)^\ell$, then
\begin{enumerate}
\item $\calC' \setminus \calC$ has a~$c_\ell$-buffer in~$I$ and
\item $\calC'$ covers at most~$A_\ell$ indices in~$I$.
\end{enumerate}

For the base case we need to
consider~$\calC' = \textsc{Adversary}(\calC,I,\ell_0)$. Note that by
construction~$\calC$ does not cover any index in~$I$. Hence if the
resulting~$\calC'$ covers any index in~$I$, then this is due to a
trivial bracket pair~$P$ either coloured red or blue. Since for our
choice of
parameters~$\lvert I \rvert \geq n/(4w)^\ell \gg 8dw = 2c_{\ell_0}$
the bracket pair~$P$ is appropriately coloured.
Hence~$\calC'\setminus \calC$ has a~$c_{\ell_0}$-buffer
in~$I$. Since~$\calC'$ covers at most~$1 \leq 3d = A_{\ell_0}$ indices
in~$I$ the base case is established.

We may thus assume the above statement for adversaries of
recursion--depth~$\ell+1$. It remains to show the statement
for~$\textsc{Adversary}(\calC,I,\ell)$ of recursion--depth~$\ell$. By
an inner induction on~$i = 0, 1, \ldots, d$ we intend to establish that
if~$\lvert I \rvert \geq n/(4w)^\ell$, then
\begin{enumerate}
\item $\tilde \calC_i \setminus \calC$ has a~$c_\ell$-buffer in~$I$
  and
  \label[prop]{prop:buffer}
\item $\tilde \calC_i$ covers at most~$w\big(A_{\ell+1} + 3i\big)$
  indices in~$I$.
  \label[prop]{prop:cover}
\end{enumerate}
Observe that in order to establish \nameref{lem:cover} it suffices to
complete the above (inner and outer) induction.

There is nothing to show for the (inner) base case~$i=0$
since~$\tilde \calC_0 = \calC$ and~$\calC$ covers no index in~$I$ by
construction.

We may thus assume the statement for~$i$ to show it for~$i+1$. First
we appeal to the \nameref{lem:large} to conclude
that~$\lvert I_i\rvert \geq n/(4w)^{\ell+1}$. This allows us to appeal
to the outer induction hypothesis to conclude
that~$\tilde \calC_{i+1} \setminus \calC_i$ has a~$c_{\ell+1}$ buffer
in~$I_i$.

To establish \cref{prop:buffer} we may appeal to the
\nameref{lem:buffer} with the inner and outer induction hypothesis.

It remains to establish \cref{prop:cover}. According to
\cref{prop:container-size} there are at most~$w$ bracket pairs in the
configuration~$\tilde \calC_{i+1}$. Consider one such bracket
pair~$P \in \tilde \calC_{i+1}$ contained in~$\area(P) \subseteq
I$. By \cref{lem:single-legacy} we may assume that there is
a~$j \leq i+1$ such that all ancestors are legacy-$j$ bracket pairs.
By induction we know that~$\lvert I_j \rvert \geq n/(4w)^{\ell+1}$ and
we may thus appeal to the outer inductive hypothesis to conclude that
the legacy ancestors of~$P$ cover at most~$A_{\ell+1}$
indices. Combining this estimate with \cref{lem:area-decompose} we
obtain that~$\lvert \area(P)\rvert \leq A_{\ell+1} + 3(i+1)w$.  This
implies
that~$\lvert \area(\tilde \calC_{i+1}) \cap I \rvert \leq
w\big(A_{\ell+1} + 3(i+1)\big)$ since the set~$U$ (as in
\cref{lem:area-decompose}) is independent of the choice of the bracket
pair. This establishes \cref{prop:cover} of the inner induction
for~$i+1$.

This establishes the induction and thereby the
\nameref{lem:cover}. \Cref{thm:depth-lb} follows from combining
\cref{lem:fail,lem:cover,lem:large} and the fact
that~$\textsc{Adversary}(\emptyset,[n],0)$
plays~$d^{\ell_0} = n^{\Omega(\frac{\log n}{\log w})}$ rounds of the
prover--adversary game before terminating.

\section{Proofs of \texorpdfstring{\cref{thm:main-circuit,thm:main-proof}}{Theorems 1 and 2}}
\label{sec:lifting}

Before proving \cref{thm:main-proof,thm:main-circuit} we need to
recall what it means to compose a CNF formula~$F$ with indexing
gadget~$\Ind_t\colon [t] \times \{0, 1\}^t \to \{0, 1\}$ that maps an
input~$(x, y)$ to $y_x$.
If~$z_1, z_2, \dots, z_n$ denote the variables of~$F$,
then~$F \circ \Ind_t^n$ denotes the CNF formula obtained from~$F$ by
substituting each variable~$z_i \mapsto \Ind_t(x_i, y_i)$,
where~$x_i, y_i$ are new sets of variables. Note that if~$F$ is an
unsatisfiable $k$-CNF formula with~$m$ clauses,
then~$F \circ \Ind_t^n$ is an unsatisfiable CNF formula over~$O(nt)$
variables and~$O(tm^k + n)$ clauses.

Let us start with the proof of \cref{thm:main-proof}. The theorem
follows from \cref{thm:width-ub,thm:depth-lb} along with an
application of dag-like lifting \cite{Garg20}. The following theorem
is taken from~\cite{Fleming22}.

\begin{theorem}[\cite{Garg20, LovettMMPZ22, Fleming22}]
  \label{th:lifting-proofs}
  Let $\varepsilon > 0$ be any constant and let $F$ be an
  unsatisfiable CNF formula on $n$ variables. For any semantic CP
  proof $\Pi$ of $F \circ \Ind^n_{n^{1 + \varepsilon}}$ there is a
  resolution refutation $\Pi^*$ of $F$ satisfying:
  \begin{itemize}
  \item
    $\mathrm{width}(\Pi^*) =
    O\left( \frac{\log |\Pi|}{\log n}\right)$;
  \item
    $\mathrm{depth}(\Pi^*) =
    O\left( \mathrm{depth}(\Pi) \frac{\log|\Pi|}{\log n} \right)$.
  \end{itemize}   
\end{theorem}

The width and size of refuting a CNF formula~$F$ are closely related to
the size of refuting the composed
formula~$F \circ \Ind^n_{n^{1+\eps}}$. The following claim is folklore
-- the proof is straightforward and may be found in,
e.g.,~\cite{SegerlindBI04}.

\begin{claim}[\cite{SegerlindBI04}, \cite{Garg20}]
  \label{cl:liftedres-up}
  Let~$F$ be an unsatisfiable CNF formula. If~$F$ admits a resolution
  refutation of width~$w$ and size~$s$, then~$F \circ \Ind_t^n$ admits
  a resolution refutation of size~$s t^{O(w)}$.
\end{claim}

With \cref{th:lifting-proofs,cl:liftedres-up} at hand we are ready to
prove \cref{thm:main-proof}. Let us restate it here for convenience.

\TheoremMainProof*

Let us remark that the lower bound in \cref{thm:main-proof} even holds
for the semantic Cutting Planes proof system.

\begin{proof}
  Let~$\varepsilon > 0$ and
  choose~$F \coloneqq \ExBracket \circ \Ind^n_{n^{1 +
      \varepsilon}}$. Recall that the formula~$\ExBracket$ is a
  $3$-CNF formula with~$\poly(n)$ clauses and~$\poly(n)$
  variables. This implies in particular that~$F$ contains~$\poly(n)$
  clauses and is defined over~$\poly(n)$ variables.
  
  Note that~\ref{it:res-ub} follows immediately from
  \cref{thm:width-ub} along with \cref{cl:liftedres-up}. We
  show~\ref{it:res-lb} by contradiction: suppose there is a resolution
  refutation of~$F$ of depth~$n^{\gamma}$ and
  size~$2^{n^{\delta}}$. By \cref{th:lifting-proofs} there exists a
  resolution proof~$\Pi$ of~$\ExBracket$ of
  width~$n^{\delta} / \log n$ and
  depth~$n^{\gamma + \delta} / \log n$. This is in contradiction to
  \cref{thm:depth-lb} for small~$\delta$.
\end{proof}

Let us turn our attention to \cref{thm:main-circuit}, that is, our
main result regarding monotone circuits. The proof follows by similar
arguments as the proof of \cref{thm:main-proof}. We again apply a
dag-like lifting theorem \cite{Garg20}. The below statement appeared
in this form in \cite{Fleming22}.

\begin{theorem}[\cite{Garg20, LovettMMPZ22, Fleming22}]
  \label{th:lifting-ckt}
  Let~$F$ be an unsatisfiable $k$-CNF formula on~$n$ variables and~$m$
  clauses and let~$\varepsilon > 0$ be constant. There is a monotone
  Boolean function $f_F$ on $m n^{k(1 + \varepsilon)} 2^k$ variables
  such that any monotone circuit $C$ computing $f_F$ implies a
  resolution refutation $\Pi$ of $F$ satisfying:
  \begin{itemize}
  \item $\mathrm{size}(\Pi) = O\left( \mathrm{size}(C) \right)$;
  \item
    $\mathrm{depth}(\Pi) = O\left( \mathrm{depth}(C) \frac{\log
        \mathrm{size}(C)}{\log n} \right)$.
  \end{itemize}
  Any resolution refutation $\Pi^*$ of $F$ implies monotone
  circuit~$C$ for $f_F$ of size~$n^{O(\mathrm{width}(\Pi^*))}$.
\end{theorem}

The monotone function~$f_F$ can be extracted in several ways
\cite{Razborov90, Gal98, Goos18, HrubesP17} from a CNF formula~$F$. We
refer the interested reader to \cite{Robere18} for a detailed
discussion. Let us restate our main result regarding monotone
circuits.

\TheoremMainCkt*

\begin{proof}
  We use the formula~$F \coloneqq \ExBracket$ and by
  \cref{th:lifting-ckt} we create a monotone function $f_F$. The
  formula $\ExBracket$ is a $3$-CNF formula with~$\poly(n)$ variables
  and clauses. Hence~$f_F$ is defined over~$\poly(n)$ variables

  According to \cref{thm:width-ub} there is a width~$O(\log n)$
  resolution refutation of~$\ExBracket$. Hence~\ref{it:ckt-ub} readily
  follows by \cref{th:lifting-ckt}. 
  We show~\ref{it:ckt-lb} by contradiction: suppose there is a
  monotone circuit computing~$f_F$ in depth~$n^{\gamma}$ and size
  $2^{n^{\delta}}$. By \cref{th:lifting-ckt} there exists a resolution
  proof $\Pi$ of $\ExBracket$ of size $2^{n^{\delta}}$ and depth
  $n^{\gamma + \delta} / \log n$. For small~$\delta$ this contradicts
  \cref{thm:main-proof}.
\end{proof}

\appendix

\section{Omitted Proofs}
\label{sec:proof-width-index-width}

In the following we prove
\cref{prop:index-width,prop:index-width-rev}. Before diving into the
respective proofs we need to record a simple fact.

\begin{claim}
  \label{rm:extension-play}
  Denote by~$k$ the index-width of~$F$ and consider clauses $A, B$
  such that~$A \lor B$ has index-width~$\leq k$. For any clause~$D$
  there is a constant size resolution derivation of
  \begin{itemize}
  \item $D \lor y_{A \lor B}$ from $D \lor y_A \lor y_B$ and~$F'$,
  \item $D \lor y_A \lor y_B$ from $D \lor y_{A \lor B}$ and $F'$.
  \end{itemize}
\end{claim}

\begin{proof}
  Consider the first statement. Recall that~$F'$ contains clauses that
  ensure that~$y_{A \lor B} = y_A \lor y_B$. Hence~$y_{A \lor B}$
  semantically follows from these clauses of~$F'$ and the
  clause~$y_A \lor y_B$. As resolution is implicationally complete
  there is a resolution derivation of~$y_{A \lor B}$
  from~$y_A \lor y_B$ and the relevant clauses of~$F'$. Since this
  derivation is over a constant number of variables it is of constant
  size. We obtain the desired statement by replacing the
  clause~$y_A \lor y_B$ by~$D \lor y_A \lor y_B$ and propagating this
  replacement through the resolution derivation. The second statement
  follows by a similar consideration.
\end{proof}

\PropositionIndexWidth*

\begin{proof}
  Let~$\pi = (D_1, D_2, \dots, D_s)$ be a resolution refutation of the
  formula~$F$ in index-width~$w$ and depth~$d$. Note that each
  clause~$D_i\in \pi$ can be written
  as~$D_i = C_{i, 1} \lor C_{i, 2} \lor \cdots \lor C_{i, n}$
  where~$C_{i, j}$ is a clause that consists of variables related to
  index~$j$ only.

  Denote by~$\tau$ the function mapping each~$D_i \in \pi$ to the
  clause~$y_{C_{i, 1}} \lor y_{C_{i, 2}} \lor \cdots \lor y_{C_{i,
      n}}$. Let $\pi' \coloneqq (\tau(D_1), \tau(D_2), \dots, \tau(D_s))$ and
  note that every clause in~$\pi'$ has width at most~$w$. The
  sequence~$\pi'$ is \emph{not} a resolution refutation yet. In the
  following we show that any clause~$\tau(D_i)$ can be derived from
  previous clauses in width~$O(w)$ and depth~$O(dw)$ thereby
  establishing the claim.
  Let us consider two cases.
  \begin{enumerate}
  \item The clause~$D_i$ is an axiom of the index-width-$k$
    formula~$F$. Hence~$D_i$ contains variables related to at most~$k$
    distinct indices and we can thus
    write~$D_i = C_{i, j_1} \lor \cdots \lor C_{i, j_k}$. Using
    \cref{rm:extension-play}~$k$ times we obtain a derivation
    of~$y_{C_{i, j_1}} \lor \cdots \lor y_{C_{i, j_k}} = \tau(D_i)$
    from the axioms of~$F'$ where we rely on the fact
    that~$y_{D_i}\in F'$. This establishes the first case.
    
  \item The clause~$D_i$ is obtained by an application of the
    resolution rule from clauses~$A \lor x$ and~$B \lor \neg x$. In
    this case, we have already derived $\tau(A \lor x)$ and
    $\tau(B \lor \neg x)$. According to \cref{rm:extension-play} there
    is a constant sized resolution derivation of~$\tau(A) \lor y_x$
    and~$\tau(B) \lor y_{\neg x}$ from~$\tau(A \lor x)$,
    $\tau(B \lor \neg x)$ and~$F'$. From~$\tau(A) \lor y_x$
    and~$\tau(B) \lor y_{\neg x}$ and the
    axiom~$\neg y_x \lor \neg y_{\neg x}$ we can
    derive~$\tau(A) \lor \tau(B)$ by applying the resolution rule
    twice.

    We can
    write~$\tau(A) \lor \tau(B) = y_{A_1} \lor y_{A_2} \lor \dots \lor
    y_{A_n} \lor y_{B_1} \lor y_{B_2} \lor \dots \lor y_{B_n}$
    where~$A_i \subseteq A$ ($B_i \subseteq B$) is the part of~$A$
    (of~$B$) related to index $i$. Since the clauses~$A,B$ are of
    index-width at most~$w-1$ there are at most~$2w - 2$ indices in
    the above expansion. Hence by appealing at most $w$ times to
    \cref{rm:extension-play} we obtain
    $y_{A_1 \lor B_1} \lor y_{A_2 \lor B_2} \lor \dots \lor y_{A_n
      \lor B_n} = \tau(D_i)$ as required. \qedhere
  \end{enumerate}
\end{proof}

\PropositionIndexWidthRev*

\begin{proof}
  Consider a width-$w$ depth-$d$ resolution
  refutation~$\pi = (D_1, D_2, \dots, D_s)$ of the formula~$F'$. To
  show the existence of a small depth and index-width resolution proof
  of $F$ we exhibit a strategy for the prover to win the
  prover--adversary game (see \cref{sec:pa-games}) in
  memory-size~$5w$.

  A partial assignment~$\rho \in (\Sigma\cup\{*\})^n$ induces a
  partial assignment~$\mu_{\rho}$ to the variables of~$F'$
  by~$\mu_{\rho}(y_D) = D|_{\rho}$, where $\mu_\rho(y_D) = *$ if and
  only if $\rho$ maps all the variables of $D$ into $\{0, *\}$ and at
  least one variable is mapped to $*$.

  Throughout the strategy the prover maintains a clause~$D$ and a
  partial assignment~$\rho \in (\Sigma\cup\{*\})^n$ such that
  \begin{itemize}
  \item $\mu_{\rho}$ falsifies~$D$, and
  \item $\rho$ is minimal with the above property.
  \end{itemize}
  The prover attempts to trace a path in the graph of the resolution
  proof~$\pi$ starting from contradiction~$D_s = \bot$. In each round
  the prover updates~$D$ to one of the premises of~$D$ to eventually
  reach an axiom of~$F'$. Since by definition~$\mu_{\rho}$ cannot
  violate an extension axiom of~$F'$ the prover must reach a
  clause~$y_D$ for~$D \in F$. Since the assignment~$\mu_{\rho}$
  falsifies~$y_D$ if and only if~$\rho$ falsifies~$D$ the prover wins
  the game.
  
  Initially~$\rho \coloneqq *^n$. Suppose that the prover remembers
  the clause~$D_i$ and that $\rho$ satisfies the aforementioned
  properties. If~$D_i$ is derived from~$D_j = A \lor y_C$
  and~$D_k = B \lor \neg y_C$ by an application of the resolution rule
  the prover queries all symbols mentioned by~$C$ (of constant
  index-width) and updates~$\rho$ accordingly. Either~$D_j$ or~$D_k$
  is falsified by $\mu_\rho$. Update~$D_j$ to point to the falsified
  clause and shrink~$\rho$.
\end{proof}

\bigskip\bigskip
\subsubsection*{Acknowledgements}

We thank Alexandros Hollender (who declined a coauthorship) for
discussions that inspired the bracket principle.  M.G., G.M., and
D.S.\ are supported by the Swiss State Secretariat for Education,
Research and Innovation (SERI) under contract number MB22.00026. K.R.\
is supported by Swiss National Science Foundation project
\mbox{200021-184656} “Randomness in Problem Instances and Randomized
Algorithms”.

\pagebreak

\DeclareUrlCommand{\Doi}{\urlstyle{sf}}
\renewcommand{\path}[1]{\small\Doi{#1}}
\renewcommand{\url}[1]{\href{#1}{\small\Doi{#1}}}
\bibliographystyle{alphaurl}
\bibliography{references}

\end{document}